%% file: Escape_isit.tex
\documentclass[journal]{IEEEtran}
\usepackage{amssymb,stmaryrd,amsmath,amsfonts,rotating,mathrsfs, psfrag}
\usepackage[noadjust]{cite}
\usepackage{color}
\usepackage{epic}
\RequirePackage{bbm}
\input{./definition}
\allowdisplaybreaks
\begin{document}
\title{On the scaling of Polar Codes:\\ II. The behavior of un-polarized channels}
\author{ S. Hamed Hassani, Kasra Alishahi and Rudiger Urbanke\thanks{S.H. Hassani and R. Urbanke are with EPFL, School of Computer
 \& Communication Sciences, \{seyedhamed.hassani, ruediger.urbanke\}@epfl.ch. Kasra Alishahi is with the department of Mathematical sciences, Sharif university of technology, alishahi@sharif.edu. This work was supported by grant no 200021-121903 of the Swiss
National Foundation.
}
}

\maketitle

\begin{abstract}
We provide upper and lower bounds on the escape rate of the Bhattacharyya
process corresponding to polar codes and transmission over the the
binary erasure channel.  More precisely, we bound the exponent of the
number of sub-channels whose Bhattacharyya constant falls in a fixed
interval $[a,b]$. Mathematically this can be stated as bounding the
limit $\lim_{n \to \infty} \frac{1}{n} \ln \mathbb{P}(Z_n \in [a,b])$,
where $Z_n$ is the Bhattacharyya process. The quantity $\mathbb{P}(Z_n \in [a,b])$ represents the fraction of sub-channels that are still un-polarized at time $n$.  \end{abstract}
\section{Introduction and main result}
The construction of polar codes (\cite{Ari09}) is done by exploring a phenomenon called channel polarization in which from a BMS channel $W$, $N=2^n$ sub-channels $\{W_{2^n} ^{(i)}\}_{1 \leq i \leq 2^n}$ are constructed with the property that almost a fraction of $I(W)$ of them tend to become noise-less (i.e., have capacity close to $1$) and a fraction of $1-I(W)$ of them tend to become completely noisy (i.e., have capacity close to $0$). Hence, as $n$ grows large, nearly all the sub-channels are in one of the following two states: highly noisy or highly noiseless. The construction of these channels is done recursively, using a transform called channel splitting. Channel splitting is a transform which takes a BMS channel $W$ as input and outputs two BMS channels $W ^+ $ and $W^-$. We denote this transform by $W \rightarrow (W^+, W^-)$.
To analyze the behavior of the sub-channels, a probabilistic approach is introduced in \cite{Ari09} and \cite{ArT09}. In this regard, the polarization process of a BMS channel $W$, denoted by $W_n$, is defined by $W_0=W$ and
\begin{equation}
W_{n+1}=  \left\{
\begin{array}{lr}
W_n ^{+} &  ; \text{with probability $\frac 12$},\\
W_n ^{-} &  ; \text{with probability $\frac 12$}.
\end{array} \right.
\end{equation}
As a result at time $n$ the process $W_n$ uniformly and randomly outputs a sub-channel from a set of $2^n$ possible sub-channels which are precisely the sub-channels $\{W_{2^n} ^{(i)}\}_{1 \leq i \leq 2^n}$.\footnote{For more details, please refer to \cite{HU10}}  The Bhattacharyya process of channel $W$ is then defined by $Z_n=Z(W_n)$, where $Z()$ denotes the Bhattacharyya constant. It was shown in \cite{Ari09} that the process $Z_n$ is a super-martingale that converges to a random variable $Z_{\infty}$. The value of $Z_{\infty}$ is either $0$ (representing the fraction of noiseless sub-channels) or $1$ (representing the fraction of noisy sub-channels) with $\mathbb{P}^{W}(Z_{\infty}=0)=I(W)$.  We call the two values $0$ and $1$ the \textit{fixed points} of the process $Z_n$ meaning that as $n$ tends to infinity, with probability one the process $Z_n$ ends up in one of the these two fixed points. The asymptotic behavior of the process $Z_n$ around the points $0$ and $1$ has been studied in \cite{ArT09} and \cite{HU10}. However at each time $n$ there still exists a positive probability, although very small, that the process $Z_n$ takes a value not so close to the fixed points. The main objective of this paper is to study these vanishing probabilities.  More precisely, let $0<a<b<1$ be constants. The quantity   $ \mathbb{P}^{W}(Z_n \in [a,b])$ represents the probability that the value of $Z_n$ is away from the two fixed points $0$ and $1$ or in other words has \textit{escaped} from the fixed points. For a channel $W$ we define the upper escape rate $\lambda_u^W$  and the lower escape rate $\lambda_l^W$ as \footnote{All the logarithms in this paper are in base 2.}
\begin{align}
\lambda_u^{W} = \lim_{[a,b] \to (0,1)}\limsup_{n \to \infty} \frac{1}{n} \log \mathbb{P}^{W}(Z_n \in [a,b]) \label{escape_def11} \\ 
\lambda_l^{W}= \lim_{[a,b] \to (0,1)} \liminf_{n \to \infty} \frac{1}{n} \log \mathbb{P}^{W}(Z_n \in [a,b]) . \label{escape_def12}
\end{align}
It is easy to see that the above defined quantities are well defined. Also, when $\lambda_u^{W}=\lambda_l^{W}=\lambda^{W}$, we say that the escape rate of the channel $W$ exists and is equal to $\lambda^{W}$. In words, as $n$ goes large, one expects that
\begin{align*}
 2^{\lambda_l^W n} \lessapprox \mathbb{P}^{W}(Z_n \in [a,b]) \lessapprox 2^{\lambda_u^W n}.
\end{align*}
In the context of polar codes, the quantity  $ \mathbb{P}^{W}(Z_n \in [a,b])$ represents the ratio of the sub-channels that have not ``polarized`` at time $n$.
In this paper we consider the case when the channel $W$ is a binary erasure channel (BEC). In the analysis of polar codes, the analysis of  binary erasure channels is more significant than other BMS channels. This is because firstly the Bhattacharyya process $Z_n = Z(W_n)$ corresponding to a BEC channel with erasure probability $z$ (BEC($z$)) is relatively more easier to analyze and it can be described in a closed numerical form  (\cite{Ari09}) as $Z_0 =z$ and
\begin{equation} \label{process}
Z_{n+1}=  \left\{
\begin{array}{lr}
{Z_n}^2 &  ; \text{with probability } \frac{1}{2},\\
2Z_n-{Z_n}^2 &  ; \text{with probability } \frac{1}{2}.
\end{array} \right.
\end{equation}
Secondly the quantities corresponding to BEC channels often provide bounds for general BMS channels. Let the functions $p_n^{a,b}(z)$ and $\theta_n^{a,b}(z)$ be defined as \footnote{To keep things simple, instead of $\mathbb{P}^{\text{BEC($z$)}} (Z_n \in [a,b])$ we write $\mathbb{P}^{z} (Z_n \in [a,b])$.}
\begin{align}
&p_n^{a,b}(z)= \mathbb{P}^{z}(Z_n \in [a,b]), \label{p}\\
&\theta_n^{a,b}(z)= \frac{1}{n} \log p_n ^{a,b}(z) \label{theta}.
\end{align}
As a result the upper and lower escape rate for the channel BEC($z$) can be stated as
\begin{align}
\lambda_u^{\text{BEC($z$)}} = \lim_{a \to 0, b \to 1}\limsup_{n \to \infty} \theta_n^{a,b}(z)\label{escape_def111} \\
\lambda_l^{\text{BEC($z$)}}= \lim_{a \to 0, b \to 1} \liminf_{n \to \infty} \theta_n^{a,b}(z) \label{escape_def112}.
\end{align}
In the sequel, we slightly modify the definition of the escape rates given in \eqref{escape_def111} and \eqref{escape_def112} and consider the following quantities,
\begin{align}\label{escape_def1}
\lambda_{u}(z,a,b,\delta) = \limsup_{n \to \infty} \sup_{x \in [z-\delta, z+\delta] } \theta_n^{a,b}(x)  \\
\lambda_{l}(z,a,b,\delta)=  \liminf_{n \to \infty} \sup_{x \in [z-\delta, z+\delta] } \theta_n^{a,b}(x)  .
\end{align}
where $\delta \in (0,1)$ is chosen in a way that $[z-\delta, z+\delta] \subseteq (0,1)$ (we call such a pair of $(z,\delta)$ a consistent pair). In words, we allow a small perturbation, namely $\delta$, in the erasure probability of the channel and define the escape rates accordingly. Therefore, when the value of $\delta$ tends to $0$, the above quantities are a good estimate of the ones given in \eqref{escape_def111} and \eqref{escape_def112}. In this paper we first show that 
\begin{lemma}\label{lim_equal}
The value of $\lambda_{u}(z,a,b,\delta)$ and $\lambda_{l}(z,a,b,\delta) $ is the same for all choices of $a$, $b$ and  $(z,\delta)$ such that $a \leq b^2$. We denote the two values by $\lambda_{u}^{\text{BEC}}$  and $\lambda_{l}^{\text{BEC}}$   respectively. \\
\QED
\end{lemma}
Numerical simulations show that the values of upper and lower escape rate are both equal to $-0.2758$ for all the BEC channels. In this paper we provide upper and lower bounds on the values of $\lambda_{u}^{\text{BEC}}$  and $\lambda_{l}^{\text{BEC}}$.
\begin{theorem}\label{bounds}
 We have
\begin{equation}
 -0.2786 \approx \frac{1}{2 \ln 2} -1 \leq \lambda_{l}^{\text{BEC}} \leq  \lambda_{u}^{\text{BEC}} \leq -0.2669.
\end{equation}
\QED
\end{theorem}
The outline of the paper is as follows. In Section~\ref{Do}  we
introduce the basic notations, definitions and tools used in this paper.
Section~\ref{Se} contains the proof of the main results of this paper followed by
section IV that contains further proofs regarding the auxiliary lemmas stated in the paper.
\section{Definitions, notations and preliminary lemmas } \label{Do}
In this section we first give a different but entirely equivalent description of the process $Z_n$ given in \eqref{process} with the help of a collection of maps denoted by $\phi_{\omega_n}$ for $n \in \naturals$. From this we also derive the relation between the quantity  $\frac{1}{n} \log \mathbb{P}^{z}(Z_n \in [a,b])$ and the maps $\phi_{\omega_n}$. We then continue by analyzing the functions $p_n^{a,b}(z)$ and $\theta_n^{a,b}(z)$ defined in \eqref{p} and \eqref{theta} and derive the relations between the functions $\theta_n^{a,b}$ for different values of $n$ and $z$.
\subsection{analyzing the random maps $\phi_{\omega_n}$}
Let $\{B_n\}_{n \in \naturals}$ be a sequence of iid Bernoulli($\frac12$) random variables. Denote by $(\mathcal{F}, \Omega, \mathbb{P})$ the probability space generated by this sequence and let  $(\mathcal{F}_n, \Omega_n, \mathbb{P}_n)$ be the probability space generated by $(B_1, \cdots,B_n)$. We now couple the process $Z_n$ with the sequence $\{B_n\}_{n \in \naturals}$. We start by $Z_0=z$ and
\begin{equation} \label{process_B}
Z_{n+1}=  \left\{
\begin{array}{lr}
{Z_{n-1}}^2 &  ; \text{if } B_n=1,\\
2Z_{n-1}-{Z_{n-1}}^2 &   ; \text{if } B_n=0.
\end{array} \right.
\end{equation}
Also, consider the two maps $T_0,T_1:[0,1]\longrightarrow [0,1]$ defined as
\begin{equation}
T_0(x)=2x-x^2 , T_1(x)=x^2.
\end{equation}
The value of $Z_{n}$ is obtained by applying $T_{B_i}$ on the value of $Z_{n-1}$, i.e., $Z_n = T_{B_n}(Z_{n-1})$. The same rule applies for obtaining the value of $Z_{n-1}$ form $Z_{n-2}$ and so on. Thinking this through recursively, the value of $Z_n$ is obtained from the starting point of the process,  $Z_0=z$, via the following maps.
\begin{definition}\label{phi}
For each $n \in \naturals$ and a realization $(b_1, \cdots,b_n) \triangleq \omega_n \in \Omega_n$ define the map $\phi_{\omega_n}$ by
\begin{equation}\nonumber
\phi_{\omega_n}=T_{b_n} \circ T_{b_{n-1}} \circ \cdots T_{b_1}.
\end{equation}
Let $\Phi_n$ be the set of all such $n$-step maps. Thus each $\phi_{\omega_n} \in \Phi_n$ is with a one-to-one correspondence with a realization  $(b_1, \cdots, b_n)$ of $ \Omega_n $.
\end{definition}
As a result, an equivalent description of the process $Z_n$ is as follows. At time $n$ the value of $Z_n$ is obtained by picking uniformly at random one of the functions in $\phi_{\omega_n} \in \Phi_n$ and assigning the value  $ \phi_{\omega_n}(z)$ to $Z_n$. Consequently we have,
\begin{align}\label{equ:equivalent}
 \mathbb{P}^z(Z_n \in [a,b]) &= \sum_{\phi_{\omega_n} \in \Phi_n}\frac{1}{2^n} \mathbb{I} (\phi_{\omega_n}(z) \in [a,b]) \\
& = \sum_{\phi_{\omega_n} \in \Phi_n}\frac{1}{2^n} \mathbb{I}(z \in \phi_{\omega_n} ^{-1}([a,b])). \nonumber
\end{align}
Therefore, in order to analyze the behavior of the quantity $\frac{1}{n} \log  \mathbb{P}^z(Z_n \in [a,b]) $ as $n$ grows large, characterizing the asymptotic behavior of the random maps $\phi_{\omega_n}$ is necessary. Continuing the theme of Definition~\ref{phi}, one can correspond to each realization of the infinite sequence $\{B_n\}_{n \in \naturals}$, denoted by $\{b_n\}_{n \in \naturals}$, a
sequence of maps $\phi_{\omega_1}(z), \phi_{\omega_2}(z),\cdots$, where $\omega_i\triangleq (b_1, \cdots,b_i)$. We call the sequence
$\{\phi_{\omega_k}\}_{k \in \naturals}$ the corresponding sequence of maps for
the realization $\{b_k\}_{k \in \naturals}$. We also use the
realization $\{b_k\}_{k \in \naturals}$ and its corresponding
$\{\phi_{\omega_k}\}_{k \in \naturals}$ interchangeably.
We now focus more on the asymptotic characteristics of the functions $\phi_{\omega_n}$. Firstly, since $\phi_{\omega_n}(z)$  has the same law as $Z_n$ starting at
$z$, we conclude that for $z \in (0,1)$ with
probability one, the quantity $\lim_{k \to \infty} \phi_{\omega_k} (z)$ takes on a value in the set $\{0,1\}$ . In Figure \ref{sample_threshold} the the functions $\phi_{\omega_n}$ are plotted for a random realization. As it is apparent from Figure~\ref{sample_threshold}, the functions $\phi_{\omega_n}$ seem to converge point-wise to a step function. This is justified in the following lemma.
\begin{figure}[htb]
\centering
 \includegraphics[width=9cm]{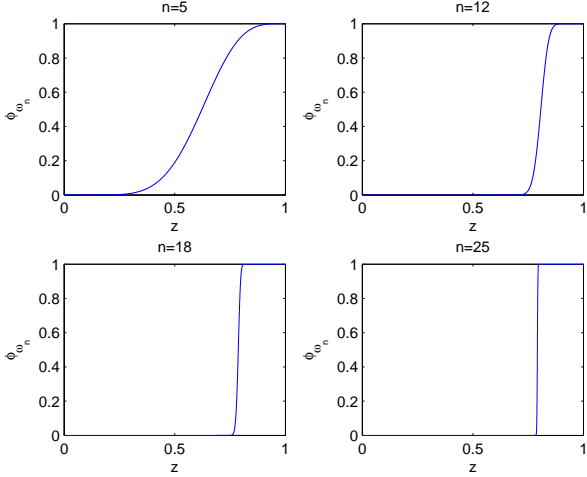}
 \caption{The functions $\phi_{\omega_n}$ associated to a random realization are plotted. As we see as $n$ grows large, the functions $\phi_{\omega_n}$ converge point-wise to a step function.}
 \label{sample_threshold}
 \end{figure}
\begin{lemma}[Almost every realization has a threshold point]\label{threshold}
For almost every realizations of $\omega \triangleq \{b_k\}_{k \in \naturals} \in \Omega$, there exists a point
$z_{\omega}^* \in [0,1]$, such that
\begin{equation}\nonumber
\lim_{n \to \infty} \phi_{\omega_n}(z) \rightarrow  \left\{
\begin{array}{lr}
0  & z \in [0,z_{\omega}^*) \\
1 &  z \in (z_{\omega}^*,1]
\end{array} \right.
\end{equation}
Moreover, $z_{\omega}^*$ has uniform distribution on $[0,1]$. We call the point $z_{\omega}^*$ the threshold point of the realization $\{b_k\}_{k \in \naturals}$ or the threshold point of its corresponding sequence of maps $\{\phi_{\omega_k}\}_{k \in \naturals}$.
\\ \QED
\end{lemma}
Looking more closely at \eqref{equ:equivalent}, by the above lemma we conclude that as $n$ grows large, the maps $\phi_{\omega_n}$ that activate the identity function $\mathbb{I}(.)$ must have their threshold point sufficiently close to $z$.
\subsection{Properties of the functions $\theta_n$ and $p_n$}
In this part we focus on the asymptotic value of functions $\theta_n ^{a,b}(z)$ and $p_n^{a,b}(z)$ given by \eqref{p} and \eqref{theta}. The following lemma states that the choice of $a$ and $b$ is not important.
\begin{lemma} [Equality of the limsups and liminfs] \label{equal}
For  two intervals $[a,b] \text{, } [c,d] \in (0,1)$,
 such that $a \leq b^2$ and $c \leq d^2$ and for a consistent pair $(z, \delta)$ we have $$\lambda_{u}(z,a,b,\delta)=\lambda_{u}(z,c,d,\delta),$$ and $$\lambda_{l}(z,a,b,\delta)=\lambda_{l}(z,c,d,\delta) .$$
\QED
\end{lemma}
Therefore, without loss of generality we can fix the value of $a$ to $\frac 14$ and the value of $b$ to $\frac 34$ and prove all the statements that appear in the sequel assuming this specific choice of $a$ and $b$. However by Lemma~\ref{equal} there is no loss of generality in the original statement of the main results of the paper. Also, in the sequel $a$ and $b$ represent this specific choice mentioned above and we will drop the superscripts $a,b$ whenever it is
clear from the context.
\begin{lemma}[Inequalities between the functions $\theta_n$] \label{inequality}
 For $n \in \naturals$ and $z \in (0,1)$ we have
\begin{enumerate}
 \item[(a)]
\begin{equation*}\label{bound3}
  {\theta}_{n+1} (z) + \frac{1}{n+1}\geq  {\theta}_{n}(z).
\end{equation*}
\item[(b)]
\begin{equation*} \label{bound4}
 \theta_{n+1}(z) + \frac{1}{n+1} \geq \text{max} \{ \theta_n(z^2), \theta_n (2z-z^2) \}.
\end{equation*}
\end{enumerate}
\QED
\end{lemma}
Lemma~\ref{inequality} relates the values of the functions $\theta_n$ on different point of the interval $(0,1)$ together. The result of Lemma~\ref{inequality} can be formalized more generally in the following way. We first define the sets $F_z^n$ and $B_z^n$ for $z \in (0,1)$ and $n \in \naturals$. These sets and their asymptotic properties are among the main tools in proving the main results.
\begin{definition}\label{expand}
Let $z \in (0,1)$. Let $ F_{z} ^{n} = \{\phi_{\omega_k}(z) \mid k \leq n, \phi_{\omega_k}
\in \Phi_k \}$ and $ B_{z} ^{n} = \{\phi_{\omega_k} ^{-1}(z) \mid k \leq n,
\phi_{\omega_k} \in \Phi_k \}$. We call the sets $F_{z} ^n$ and $B_{z} ^n$ the
$n$th forward and backward sets due to $z$. Further we call the sets
$F_{z}= \cup_{n} F_n ^{z}$ and $B_z=\cup_{n} B_n ^{z}$ the forward
and backward sets due to $z$. In general for an arbitrary set $A \in
(0,1)$, by the forward set due to $A$, denoted by $F_A$, we mean $F_A =
\bigcup_{z \in A} F_z $. The backward set due to $A$, denoted by $B_A$,
is defined similarly.
From Lemma~\ref{inequality} we can easily conclude the following.
\end{definition}
\begin{corollary} \label{inequalities}
Let $z \in (0,1)$ and $m,n \in \naturals$.
\begin{enumerate}
 \item For $x \in F_z^{m} $ we have
\begin{equation*}
 \theta_{n+m}(z) + \frac{m}{n} \geq \theta_n (x).
\end{equation*}
\item For $y \in B_z^m$ we have
\begin{equation*}
 \theta_{n+m}(y) + \frac{m}{n} \geq \theta_n (z).
\end{equation*}
\end{enumerate}
\QED
\end{corollary}
\section{Proof of the main results} \label{Se}
\subsection{Proof of Lemma~\ref{lim_equal}}
Consider the sequence $\{a_n\}_{n \in \naturals}$ defined as
\begin{equation}
 a_n := \sup_{z \in [\frac 14, \frac 34]} \theta_n(z).
\end{equation}
 We claim that for any consistent pair $(z, \delta)$, we have
\begin{align}\label{escape_def1}
\lambda_{u}(z,\frac 14, \frac 34, \delta) = \limsup_{n \to \infty} a_n  \\
\lambda_{l}(z,\frac 14, \frac 34, \delta)=  \liminf_{n \to \infty} a_n .
\end{align}
Clearly the above statement together with Lemma~\ref{equal} complete the proof of Lemma~\ref{lim_equal}. To prove the claim we use the following lemma.
\begin{lemma}\label{jump}
Let $[c,d]$ and $[e,f]$ be non-empty intervals in $(0,1)$. There exist a $m \in \naturals$ such that for $x \in [c,d]$ we have $B_x^m \cap [e,f] \neq \emptyset$.
\end{lemma}
Now fix a pair $(z,\delta)$ and let the sequence $\{u_n\}_{n \in \naturals}$ be defined as
\begin{equation*}
 u_n = \sup_{x \in [z-\delta, z+\delta] } \theta_n(x).
\end{equation*}
By Lemma~\ref{jump} there exists a $m \in \naturals$ such that for $x \in [z-\delta, z+\delta] $ we have $B_x^m \cap [\frac 14, \frac 34] \neq \emptyset$. As a result, by Corollary~\ref{inequalities} part (b) for $n \in \naturals$ we have
\begin{equation} \label{e_1}
a_{n+m} \geq u_n - \frac{m}{n}.
\end{equation}
Similarly as above, there exists a $k \in \naturals$ such that  for $n \in \naturals$
\begin{equation} \label{e_2}
u_{n+k} \geq a_n - \frac{k}{n},
\end{equation}
and the claim can easily be followed from  \eqref{e_1} and \eqref{e_2}.
\subsection{Proof of Theorem~\ref{bounds}}
\subsubsection{Lower bound} We first consider the average of the functions $p_n (z)$ over $(0,1)$ and use it to provide bounds for $\lambda_l^{\text{BEC}}$. More precisely, let the sequence $\{b_n\}_{n \in \naturals}$ be defined by
\begin{equation}\label{b_n}
 b_n := \frac{1}{n} \log[ \int_{0}^{1} \mathbb{P}^{z} (Z_n \in [a,b]) dz].
\end{equation}
 We have
\begin{lemma} \label{bound_average}
 $\lambda_l^{\text{BEC}} \geq \liminf_{n \to \infty} b_n$.
\\ \QED
\end{lemma}
We now proceed by finding a lower bound on the quantity $\liminf_{n \to \infty} b_n$. By \eqref{equal} we have:
\begin{align*}
 \int_{0} ^{1} \mathbb{P} ^{z}(Z_n \in [a,b] ) dz &=  \int_{0} ^{1} [\sum_{\phi_{\omega_n}} \frac{1}{2^n} \mathbb{I} (z \in  \phi_{\omega_n} ^{-1} [a,b])] dz \\
&=  \sum_{\phi_{\omega_n}} \frac{1}{2^n} [\int_{0} ^{1}  \mathbb{I} (z \in  \phi_{\omega_n} ^{-1} [a,b]) dz]\\
&= \mathbb{E} \vert \phi_{\omega_n} ^{-1} [a,b] \vert.
\end{align*}
Thus by taking $\frac{1}{n}\log()$ from both sides we have:
\begin{eqnarray} \label{jensen}
b_n = \frac{1}{n}\log \int_0 ^1  \mathbb{P} ^{z}(Z_n \in [a,b] ) dz & = & \frac{1}{n} \ln \mathbb{E} \vert \phi_{\omega_n} ^{-1} [a,b] \vert \\ \nonumber
& \geq & \mathbb{E} \frac{1}{n} \log \vert \phi_{\omega_n} ^{-1} [a,b] \vert   \nonumber
\end{eqnarray}
The value of $\lim_{n \to \infty}\mathbb{E} \frac{1}{n}
\ln \vert \phi_{\omega_n} ^{-1} [a,b] \vert $ is computed by the following lemma.
\begin{lemma} \label{log_measure}
 We have:
\begin{equation*}
 \lim_{n \rightarrow \infty} \mathbb{E} \frac{1}{n} \log \vert \phi_{\omega_n} ^{-1} [a,b] \vert = \frac {1}{2 \ln 2} - 1.
\end{equation*}
\QED
\end{lemma}\
As a result of the above lemma and \eqref{jensen} we have
\begin{equation*}
\lambda_{l}^{\text{BEC}} \geq \liminf_{n \to \infty} b_n \geq  \frac {1}{2 \ln 2} - 1.
\end{equation*}
\subsubsection{Upper bound} Let the process $Q_n$ be defined as $Q_n=\sqrt{Z_n(1-Z_n)}$. Following the lead of \cite[Lemma 1]{ArT09}, we have
 \begin{equation*}
 Q_{n+1}= Q_n . \left\{
\begin{array}{lr}
\sqrt{Z_n(1+Z_n)} &  ; \text{if $B_n=1$},\\
\sqrt{(2-Z_n)(1-Z_n)}	& ; \text{if $B_n=0$}.
\end{array} \right.
\end{equation*}
As a result,
\begin{align*}
 & \mathbb{E}[Q_{n+1}  \mid Q_n] \\
 & \leq \frac {Q_n}{2} \max_{z\in [0,1]} \{ \sqrt{(2-z)(1-z)}+ \sqrt{z(1+z)} \}\\
& \leq Q_n  \frac{\sqrt{3}}{2}.
\end{align*}
Thus by noting that $\mathbb{E}(Q_0) \leq 1$ we get
\begin{equation} \nonumber
\mathbb{E}({Q_n}) \leq {(\frac {\sqrt{3}}{2})}^n.
\end{equation}
Hence by the Markov inequality, it is easy to see that for $0<a<b<1$ there is some
$\alpha=\alpha(a,b)>0$ such that:
\begin{equation}\nonumber
{\mathbb{P}}^{z}(Z_n \in [a,b]) \leq \alpha ({\frac{\sqrt{3}}{4}})^{n}.
\end{equation}
Therefore, for $z \in (0,1)$
\begin{align*}\nonumber
\frac{1}{n} \log \mathbb{P}^{z}(Z_n \in [a,b])  \leq \frac 12 \log \frac{3}{4} + \frac{\log \alpha}{n}.
\end{align*}
Hence by tending $n$ to infinity we get
\begin{align*}\nonumber
\limsup_{n \to \infty}\frac{1}{n} \log \mathbb{P}^{z}(Z_n \in [a,b])  \leq \frac 12 \log \frac{3}{4}.
\end{align*}
The above idea can be generalized in the following way: let $\alpha,\beta \geq 0$ and define $Q_n=Z_n^{\alpha}(1-Z_n)^{\beta}$. Going along the same lines as above, we get
\begin{equation*}
Q_{n+1}= Q_n . \left\{
\begin{array}{lr}
Z_n^{\alpha}(1+Z_n)^{\beta} &  ; \text{if $B_n=1$},\\
(2-Z_n)^{\alpha}(1-Z_n)^{\beta}	& ; \text{if $B_n=0$}.
\end{array} \right.
\end{equation*}
Let $\lambda(\alpha, \beta)$ be defined as
\begin{equation}
 \zeta(\alpha, \beta) =\frac 12 \max_{z \in [0,1]} \{z^{\alpha}(1+z)^{\beta}+ (2-z)^{\alpha}(1-z)^{\beta}\}.
\end{equation}
We have
\begin{equation*}
 \mathbb{E}(Q_n)\leq {\zeta(\alpha,\beta)}^{n}.
\end{equation*}
And as a result
\begin{align*}\nonumber
\limsup_{n \to \infty}\frac{1}{n} \log \mathbb{P}^{z}(Z_n \in [a,b]) \leq \log \zeta(\alpha,\beta).
\end{align*} 
Minimizing the value of $\zeta(\alpha,\beta)$ over all the values of $\alpha$ and $\beta$, we get
\begin{align*}\nonumber
\limsup_{n \to \infty}\frac{1}{n} \log \mathbb{P}^{z}(Z_n \in [a,b]) \leq -0.2669.
\end{align*} 
\section{Appendix}
\subsection{ Proof of Lemma~\ref{threshold}}
Recall that for a realization $\omega = \{b_k\}_{k \in \naturals} \in \Omega$ we define $\omega_n = (b_1, \cdots,b_n)$. The maps $T_0$ and $T_1$ and hence the maps $\phi_{\omega_n}$s are increasing on $[0,1]$. Thus
$\phi_{\omega_n}(z) \rightarrow 0$ implies that $\phi_{\omega_n}(z') \rightarrow 0$ for $z'
\leq z$ and $\phi_{\omega_n}(z) \rightarrow 1$ implies that $\phi_{\omega_n}(z') \rightarrow
1$ for $z' \geq z$. Moreover, we know that for almost every $z\in (0,1)$,
$\lim_{n \to \infty} \phi_{\omega_n} (z)$ is either $0$ or $1$ for almost every realization $\{\phi_{\omega_n}\}_{n \in \naturals}$. Hence it suffices to let
\begin{equation}\nonumber
 z_{\omega}^*=\inf \{ z: \phi_{\omega_n}(z)\rightarrow 1 \}.
\end{equation}
To prove the second part of the lemma, notice that
\begin{align*}
z &= \mathbb{P}^{z}(Z_{\infty}=1)\\
&= \mathbb{P}^{z}(\phi_{\omega_n}(z)\rightarrow 1 )\\
&= \mathbb{P}^{z}(\inf \{ z: \phi_{\omega_n}(z)\rightarrow 1 \} \leq z ) \\
&= \mathbb{P}^{z}(z_{\omega}^* < z).
\end{align*}
Which shows that $z_{\omega}^*$ is uniformly distributed on $[0,1]$.
\subsection{Proof of Lemma~\ref{equal}}
Using \eqref{equ:equivalent}, we can write $p_{n+m} ^{a,b}$ as follows:
\begin{align*} \nonumber
p_{n+m} ^{a,b}(z) &= \sum_{\phi_{\omega_{n+m}}} \frac{1}{2^{n+m}} \mathbb{I}(z \in \phi_{\omega_{n+m}} ^{-1} [a,b])\\
&=\sum_{\phi_{\omega_m}} \frac{1}{2^m} \sum_{\phi_{\omega_n}} \frac{1}{2^n} \mathbb{I}(z \in \phi_{\omega_n} ^{-1}(\phi_{\omega_m} ^{-1} [a,b]))\\
&=\sum_{\phi_{\omega_m}} \frac{1}{2^m} \sum_{\phi_{\omega_n}} \frac{1}{2^n} \mathbb{I}(z \in \phi_{\omega_n} ^{-1}[\phi_{\omega_m} ^{-1} (a),\phi_{\omega_m} ^{-1}(b)]).
\end{align*}
Thus by the union bound we get
\begin{align*}
2^m p_{n+m} ^{a,b}(z) \geq  \mathbb{P}^z (Z_n \in \bigcup_{\phi_{\omega_m}} [\phi_{\omega_m}^{-1}(a),\phi_{\omega_m}^{-1}(a)] ). 
\end{align*}
now since $a \leq b^2$, it can easily be verified that $\bigcup_{\phi_{\omega_m}} [\phi_{\omega_m}^{-1}(a),\phi_{\omega_m}^{-1}(a)]$ contains a closed interval which as $m$ grows large, its Lebesgue measure approaches one. As a result there exits a $k \in\naturals$ such that $[c,d] \subseteq \bigcup_{\phi_{\omega_k}} [\phi_{\omega_k}^{-1}(a),\phi_{\omega_k}^{-1}(a)]$ and as a result for $n \in \naturals$ we have
\begin{align*}
2^k p_{n+k} ^{a,b}(z) \geq  p_{n} ^{c,d}(z). 
\end{align*}
and since $\theta_{n} ^{a,b} (z) = \frac{1}{n} \log p_{n} ^{a,b}$, we easily get
\begin{align}\label{e1}
 \theta_{n+k} ^{a,b}(z) + \frac{k}{n+k} \geq   \theta_{n} ^{c,d}(z). 
\end{align}
Similarly, since $c \leq d^2$ there exists a $l \in \naturals$ such that for $n \in \naturals$ we have
\begin{align}\label{e2}
 \theta_{n+l} ^{c,d}(z) + \frac{l}{n+l} \geq   \theta_{n} ^{a,b}(z). 
\end{align}
Now the proof of the lemma follows by \eqref{e1}, \eqref{e2} and tending $n$ to infinity.
\subsection{Proof of Lemma~\ref{jump}}
We first need the following lemma.
\begin{lemma} [Denseness of the forward and backward sets] \label{non-empty}
Let $(a,b) \subseteq (0,1)$ be a non-empty interval and $z\in (0,1)$,
\begin{enumerate}
\item [(a)] For $z \in (0,1)$ the set $B_z$ is dense in $[0,1]$.
\item [(a)] Assuming $(a,b) \subseteq (0,1)$ is a non-empty interval, the set $U^{a,b}=\cup_{n \in \naturals} \cup_{\phi_{\omega_n } \in \Phi_n} \phi_{\omega_n}^{-1}(a,b)$ is a dense and open subset of $(0,1)$.
\item[(c)] The set of points $z \in (0,1)$ for which the  set $F_z$ is dense in $(0,1)$, is a dense subset of $(0,1)$.
\end{enumerate}
\end{lemma}
\begin{proof}
For part (a), let $(c,d)$ be a non-empty interval in $(0,1)$.  We must find a function $\phi_{\omega_l} \in \Phi_l$ such
that $\phi_{\omega_l} ^{-1}(z) \in (c,d)$ or equivalently $z \in \phi_{\omega_l}(c,d)$. But
as $(c,d)$ is non-empty and the set of threshold points is dense in
$(0,1)$, there exists a threshold point $z_{\omega}^* \in (c,d)$. Let $\{\phi_{\omega_n}\}$
be the realization which corresponds to $z_{\omega}^*$. Since $\phi_{\omega_n}(c) \to 0$
and $\phi_{\omega_n}(d) \to 1$, there exists some member of this realization,
namely $\phi_{\omega_l}$, such that $z \in \phi_{\omega_l} (c,d)$.
This completes the proof of part (a). The proof of part (b) follows from part (a) and the fact that the set $U^{a,b}$ is an countable union of open sets. To prove part (c), Consider the set
\begin{equation*}
A= \bigcap_{\substack{a,b \in \mathbb{Q} \cap (0,1)\\ a < b}} U^{a,b} ,
\end{equation*}
 where by $\mathbb{Q}$ we mean the set of rational numbers. For each $z \in A$ the set $F_z$ is dense in $(0,1)$. According to part (b), all the sets $U^{a,b}$ are dense and open in $(0,1)$ . As a result, since $[0,1]$ is a compact space, the set A is also dense in $(0,1)$ by the Baire category theorem.
\end{proof}
Let $z \in [c,d]$. According to Lemma~\ref{non-empty} part (a), since $B_z$ is dense in $(0,1)$,  there exists a $\phi_{\omega_{l_z}} \in \Phi_{l_z}$ such that $\phi_{\omega_{l_z}}^{-1}(z) \in (e,f)$. Now since the function $\phi_{\omega_{l_z}}^{-1}$ is continuous then there exists a neighborhood $U_z$ around $z$ such that  $\phi_{\omega_{l_z}}^{-1}(U_z) \in (e,f)$  and as a result for $n \geq l_{z}$ and $y \in U_z$ we have $B_y^{l_z} \cap [e,f]  \neq \emptyset $. Also, since $[c,d] \subseteq \cup_{z \in [c,d]} U_z$ and $[c,d]$ is compact, then there exist $z_1, \cdots, z_l \in [c,d]$ such that $[c,d] \subseteq \cup_{i=1}^{l} U_{z_i}$. The result now follows by letting $m = \text{max}_{1 \leq i \leq l} l_{z_i}$.
\subsection{Proof of Lemma~\ref{inequality}}
For part (a) we have
\begin{align*}
\mathbb{P}^{z}(Z_{n+1} \in [a,b]) &= \frac 12 \mathbb{P}^{z^2}(Z_{n} \in [a,b])\\
& \quad \quad \quad +\frac12 \mathbb{P}^{2z-z^2}(Z_{n} \in [a,b])].
\end{align*}
Hence,
\begin{align*}
2p_{n+1}^{a,b}(z) \geq \text{max} \{p_{n}^{a,b}(z^2),p_{n}^{a,b}(2z-z^2) \},
\end{align*}
and as a result,
\begin{align*}
&\frac{1}{n+1} \log p_{n+1}^{a,b}(z) + \frac{1}{n+1}\\
& \quad \geq \frac{n}{n+1}\text{max} \{\frac{1}{n} \log p_{n}^{a,b}(z^2),\frac{1}{n} \log p_{n}^{a,b}(2z-z^2) \}\\
&\quad \geq \text{max} \{\frac{1}{n} \log p_{n}^{a,b}(z^2),\frac{1}{n} \log p_{n}^{a,b}(2z-z^2).
\end{align*}
The proof of part (a) now follows by noting that $\theta_n^{a,b}=\frac{1}{n} \log p_{n}^{a,b}$.
For part (b), using \eqref{equ:equivalent}, we can write $p_{n+1} ^{a,b}$ as follows: Let $a_1=1-\sqrt{1-a}$,  $b_1=1-\sqrt{1-b}$, $a_2=\sqrt{a}$, $b_2=\sqrt{b}$. We have
\begin{align*} \nonumber
 p_{n+1} ^{a,b}(z) &= \sum_{\phi_{n+1}} \frac{1}{2^{n+1}} \mathbb{I}(z \in \phi_{n+1} ^{-1} [a,b])\\
&=\sum_{\phi_{\omega_n}} \frac{1}{2^n} [\frac{\mathbb{I}(z \in \phi_{\omega_n} ^{-1} [a_1,b_1])+\mathbb{I}(z \in \phi_{\omega_n} ^{-1} [a_2,b_2])}{2}].
\end{align*}
Hence it is easy to see that:
\begin{equation*} \label{recursive}
p_{n+1} ^{a,b} (z)=p_n ^{a,b} (z)+ \frac{1}{2} [p_n ^{a,a_2}(z)+p_n ^{b,b_1}(z)-p_n ^{a,a_1}(z)-p_n ^{b,b_2}(z)],
\end{equation*}
or equivalently
\begin{equation} \nonumber
p_{n+1} ^{a,b}(z)= \frac 12 ( p_{n}(z) ^{a_2,b_1}+ p_{n} ^{a_1,b_2}(z) ).
\end{equation}
Now by assigning $a=\frac 14$ and $b = \frac 34$ we have $a_2 \leq b_1$. Therefore $[a,b] \subseteq [a_1, b_1]\cup [a_2, b_2]$ and
\begin{equation} \label{main}
2 p_{n+1} ^{a,b}(z) \geq  p_{n} ^{a,b}(z),
\end{equation}
Hence part (b) can be easily followed in a similar way to part (a).
\subsection{Proof of Lemma~\ref{log_measure}}
In order to compute $\lim_{n \rightarrow \infty} \mathbb{E} \frac{1}{n}
\ln \vert \phi_{\omega_n} ^{-1} [a,b] \vert$,  we define a reverse stochastic
process $\{ \bar{Z}_n \}_{n \in \naturals \cup \{ 0 \}}$ via the inverse
maps $T_0 ^{-1}$, $T_1 ^{-1}$. Pick a sequence of i.i.d. symmetric
Bernoulli random variables $B_1, B_2, \cdots$ and define
$\bar{Z}_n = \psi_{\omega_n} (z)$ where $\omega_n \triangleq (b_1, \cdots,b_n) \in \Omega_n$ and
\begin{equation}
 \psi_{\omega_n} = T_{b_n} ^{-1} \circ T_{b_{n-1}} ^{-1} \circ \cdots \circ T_{b_1} ^{-1}.
\end{equation}
\begin{lemma}
The Lebesgue measure (or the uniform probability measure) on
$[0,1]$, denoted by $\nu$, is the unique, and hence ergodic, invariant measure for the
Markov process $\bar{Z}_n$.
\end{lemma}
\begin{proof}
First note that if $\bar{Z}_n$ is distributed according to the Lebesgue measure, then
\begin{align*}
 \mathbb{P}( \bar{Z}_{n+1} < t)&=\frac{1}{2}\mathbb{P}( \bar{Z}_{n+1} < T_0(t))+ \frac{1}{2}\mathbb{P}( \bar{Z}_{n} < T_1(t))\\
&= \frac{1}{2}t^2 + \frac{1}{2} (2t-t^2)=t.
\end{align*}
This proves the invariance of the Lebesgue measure. In order to prove
the uniqueness, we will show that for any $z \in (0,1)$, $\bar{Z}_n$
converges weakly to a uniformly distributed random point in $[0,1]$, i.e.,
\begin{equation} \label{conv_Lebesgue}
 \bar{Z}_n ^{(z)}= \psi_{\omega_n}(z) \rightarrow \nu.
\end{equation}
Knowing that, uniqueness would be proved since for any invariant measure
$\rho$,
\begin{equation}
 \rho(.)= \mathbb{P} ^{\rho} (\bar{Z}_n \in .) = \int \mathbb{P} (\bar{Z}_n \in .) \rho (dz) \rightarrow \nu(.).
\end{equation}
To prove \eqref{conv_Lebesgue}, note that $\psi_{\omega_n}$  has the same
(probability) law as $\phi_{\omega_n} ^{-1}$ and we know that $\phi_{\omega_n} ^{-1}(z)
\rightarrow z_{\omega}^*$ almost surely and hence weakly but $z_{\omega}^*$ is distributed
according to $\nu$, which proves the statement.
\end{proof}
\begin{theorem} \label{log_measure}
 We have:
\begin{equation}
 \lim_{n \rightarrow \infty} \mathbb{E} \frac{1}{n} \ln \vert \phi_{\omega_n} ^{-1} [a,b] \vert = \frac 12 - \ln 2.
\end{equation}
\end{theorem}
\begin{proof}
We have:
\begin{equation} \nonumber
 \vert \psi_n [a,b] \vert= \psi_n(a)-\psi_n (b) = \psi'_n (c)(b-a),
\end{equation}
for some $c \in (a,b)$.
And by chain rule,
\begin{small}
\begin{align*}
\psi'_n (c)&= (T_{b_n} ^{-1} \circ T_{b_{n-1}} ^{-1} \circ \cdots \circ T_{b_1} ^{-1})'(c)\\
&={T_{b_1} ^{-1}}'(c). {T_{b_2} ^{-1}}'(T_{b_1} ^{-1}(c)). \cdots. {T_{b_n} ^{-1}}'(T_{b_{n-1}} ^{-1} \circ  \cdots \circ T_{\sigma_1} ^{-1}(c))\\
&={T_{b_1} ^{-1}}'(\psi_0 (c)). {T_{b_2} ^{-1}}'(\psi_1 (c)). \cdots. {T_{b_n} ^{-1}}'(\psi_{n-1}(c))).
\end{align*}
\end{small}
Or after taking logarithm,
\begin{equation}\nonumber
 \frac{1}{n} \ln (\psi '_{\omega_n} (c))=\frac{1}{n} \sum_{j=1} ^{n} \ln {T_{b_j} ^{-1}}' (\psi_{j-1}(c)).
\end{equation}
But according to the ergodic theorem, the last expression should (almost
surely) converge to the expectation of $\ln {T_{B_1} ^{-1}}'(z)$,
where $z$ is assumed to be distributed according to $\nu$. This can be
easily computed as
\begin{align*}
\mathbb{E}^{\nu}[\ln {T_{B_1} ^{-1}}'(z)]&= \frac{1}{2} \int_{0} ^{1} \ln (\sqrt{x})' dx +\frac{1}{2} \int_{0} ^{1} \ln (1-\sqrt{1-x})' dx \\
&=\frac{1}{2}- \ln 2.
\end{align*}
This completes the proof.
\end{proof}
\subsection{Proof of Lemma~\ref{bound_average}}
Define $c =  \liminf_{n \to \infty} b_n$ and let $\gamma$ be an arbitrary positive value. Our aim is to show that $\lambda_l^{\text{BEC}} \geq c- 2 \gamma$. Since  $c =  \liminf_{n \to \infty} b_n$, there exists a $K \in \naturals$ such that for $n \geq K$ we have $b_n \geq c- \gamma$. In other words for $n \geq K$ we have
\begin{equation} \nonumber
 \int_{0} ^{1} \mathbb{P} ^{z}(Z_n \in [a,b] ) dz > 2 ^ {n (c - \gamma)}.
\end{equation}
Hence for any $n >  \text{max} \{K, \frac{1}{\gamma} \}$ there exists
a $z_n \in (2 ^ {n (c - 2 \gamma)}, 1- 2 ^ {n (c - 2\gamma)})$
such that $\theta(z_n) \geq c - \gamma$. For $n > 2 \text{max} \{K, \frac{1}{\gamma} \}$ define $e_n= \lfloor n- \log_2 (-n(c-\gamma)) \rfloor$
and consider the function $\theta_{e_n}(z)$ and the particular point
$z_{e_n}$. Consider the set $B_{z_{e_n}} ^ {n-e_n} $. By Lemma~\ref{inequalities}, for any $y\in B_{z_{e_n}} ^ {n-{e_n}} $ we have:
\begin{equation} \label{equ:ff}
 \theta_{n} (y) \geq \theta_{{e_n}}(z_{{e_n}})-\frac{n-{e_n}}{{e_n}}.
\end{equation}
On the other hand, consider the functions $T_0^{-1}(z)=z^{\frac 12}$ and $T_1^{-1}(z)=1-\sqrt{1-z}$. We have
\begin{align*}
\overbrace{ T_0^{-1} \circ \cdots T^{-1}_0 } ^{n - {e_n} \text{times}}  (2^{n(c-\gamma)})  & = (2^{n(c-\gamma)})^{\frac{1}{2^{n-{e_n}}}}\\
& \geq 2^{n(c-\gamma) \times \frac{1}{-n(c-\gamma)}}\\
&  =  \frac 12.
\end{align*}
Similarly it is easy to see that if we apply $n-{e_n}$ times the function
$T_1^{-1}$ on $1- 2^{n(c-\gamma)}$, the resulting value is less than
$\frac 12$. As a result, it is easy to see that $B_{z_{e_n}} ^ {n-{e_n}}
\cap [\frac 14, \frac34] \neq \emptyset$. We further have: $\lim_{n
\to \infty} \frac{n-{e_n}}{{e_n}} \to 0$ or there exists a $K' \in \naturals$
such that for $n \geq K'$ we have $\frac{n-{e_n}}{{e_n}} < \gamma$ .  Therefore, by \eqref{equ:ff} there exits a $y_n \in
B_{z_{e_n}} ^ {n-{e_n}} \cap [\frac14, \frac 34]$ such that:
\begin{equation} \nonumber
 \theta_n(y_n) \geq c-  \gamma -\gamma.
\end{equation}
Hence for $n \geq \text{max} \{K', K\}$ we have
\begin{equation} \nonumber
\sup_{z\in [\frac14, \frac 34]} \theta_n(z)  \geq \theta_n(y_n) \geq c-  2 \gamma.
\end{equation}
\bibliographystyle{IEEEtran}
\bibliography{lth,lthpub}
\end{document}

%% file: definition.tex
\newtheorem{theorem}{Theorem}

\newcommand{\btheo}{\begin{theorem}}
\newcommand{\etheo}{\end{theorem}}
\newcommand{\bproof}{\begin{proof}}
\newcommand{\eproof}{\end{proof}}
\newtheorem{definition}[theorem]{Definition}
\newcommand{\bdefi}{\begin{definition}}
\newcommand{\edefi}{\end{definition}}
\newtheorem{fact}[theorem]{Fact}
\newcommand{\bprop}{\begin{fact}}
\newcommand{\eprop}{\end{fact}}
\newtheorem{corollary}[theorem]{Corollary}
\newcommand{\bcor}{\begin{corollary}}
\newcommand{\ecor}{\end{corollary}}
\newtheorem{example}[theorem]{Example}
\newcommand{\bex}{\begin{example}}
\newcommand{\eex}{\end{example}}
\newtheorem{lemma}[theorem]{Lemma}
\newcommand{\blemma}{\begin{lemma}}
\newcommand{\elemma}{\end{lemma}}
\newtheorem{remark}[theorem]{Remark}
\newcommand{\bremark}{\begin{remark}}
\newcommand{\eremark}{\end{remark}}
\newtheorem{conj}[theorem]{Conjecture}
\newcommand{\bconj}{\begin{conj}}
\newcommand{\econj}{\end{conj}}

\newcommand{\naturals}{\ensuremath{\mathbb{N}}}

\def\0{{\tt 0}} 
\def\1{{\tt 1}} 
\def\?{{\tt *}} 
\renewcommand{\mid}{\,|\,}



%% file: Escape_isit.bbl
\newcommand{\SortNoop}[1]{}
\begin{thebibliography}{1}
\providecommand{\url}[1]{#1}
\csname url@rmstyle\endcsname
\providecommand{\newblock}{\relax}
\providecommand{\bibinfo}[2]{#2}
\providecommand\BIBentrySTDinterwordspacing{\spaceskip=0pt\relax}
\providecommand\BIBentryALTinterwordstretchfactor{4}
\providecommand\BIBentryALTinterwordspacing{\spaceskip=\fontdimen2\font plus
\BIBentryALTinterwordstretchfactor\fontdimen3\font minus
  \fontdimen4\font\relax}
\providecommand\BIBforeignlanguage[2]{{%
\expandafter\ifx\csname l@#1\endcsname\relax
\typeout{** WARNING: IEEEtran.bst: No hyphenation pattern has been}%
\typeout{** loaded for the language `#1'. Using the pattern for}%
\typeout{** the default language instead.}%
\else
\language=\csname l@#1\endcsname
\fi
#2}}

\bibitem{Ari09}
E.~{Ar\i kan}, ``Channel polarization: A method for constructing
  capacity-achieving codes for symmetric binary-input memoryless channels,''
  \emph{IEEE Transactions on Information Theory}, vol.~55, no.~7, pp.
  3051--3073, 2009.

\bibitem{ArT09}
E.~{Ar\i kan} and E.~{Telatar}, ``{On the rate of channel polarization},'' in
  \emph{Proc. of the IEEE Int. Symposium on Inform. Theory}, Seoul, South
  Korea, July 2009, pp. 1493--1495.

\bibitem{HU10}
S.~H. Hassani and R.~Urbanke, ``On the scaling of polar codes: I. the bahavior
  of polarized bhannels,'' \emph{submitted to ISIT}, 2010.

\end{thebibliography}
